\documentclass[a4paper]{article}

\usepackage{amsmath, amsfonts, amsthm, xcolor, comment, url, hyperref}
\usepackage{graphicx}
\usepackage{multirow,array}

\theoremstyle{plain}
\newtheorem{theorem}{Theorem}
\newtheorem{lemma}[theorem]{Lemma}
\newtheorem{conjecture}[theorem]{Conjecture}
\theoremstyle{definition}
\newtheorem{definition}[theorem]{Definition}

\usepackage{tikz}

\newcommand{\note}[1]{}
\newcommand{\etal}{\emph{et al.}}

\begin{document}

\title{Conservative two-stage group testing \\ in the linear regime}
\author{Matthew Aldridge% 
\thanks{School of Mathematics, University of Leeds, UK. \url{m.aldridge@leeds.ac.uk}}}
\date{5 May 2020}

\maketitle

\newcommand{\covid}{\textsc{covid}-19}
\newcommand{\sars}{SARS-CoV-2}

\begin{abstract}
Inspired by applications in testing for \covid, we consider a variant of two-stage group testing called `conservative' (or `trivial') two-stage testing, where every item declared to be defective must be definitively confirmed by being tested by itself in the second stage. We study this in the linear regime where the prevalence is fixed while the number of items is large. We study various nonadaptive test designs for the first stage, and derive a new lower bound for the total number of tests required. We find that a first-stage design as studied by Broder and Kumar with constant tests per item and constant items per test is extremely close to optimal for all prevalences, and is optimal in the limit as the prevalence tends to zero. Simulations back up the theoretical results.
\end{abstract}

\tableofcontents

\newpage

\section{Introduction}

\subsection{Group testing}

\emph{Group testing} (or \emph{pooled testing}) is the following problem. Suppose there are $n$ individuals, some of whom are infected with a disease. If a test exists that reliably detects the disease, then each individual can be separately tested for the disease to find if they have it or not, requiring $n$ tests. However, in theory, a pooled strategy can be better: we can take samples from a number of individuals, pool the samples together, and test this pooled sample. If none of the individuals are infected, the test should be negative, while if one or more are the individuals are positive then, in theory, the test should be positive. It might then be possible to ascertain which individuals have the disease using fewer than $n$ pooled tests, thus saving resources when tests are expensive or limited.

%\bstctlcite{IEEEexample:BSTcontrol}

Experiments show that the group testing paradigm holds for \sars, the virus that causes the disease \covid; that is, pools of samples with just one positive sample and many negative samples do indeed produce positive results, at least for pools of around $32$ samples or fewer \cite{abdalhamid, benami, shental, yelin}. This work led to a great interest in group testing as a possible way to make use of limited tests for \covid; for a review of such work we point readers to the survey article \cite{review}.

Many of these papers use a similar model that we also use here: the number of individuals $n$ is large; the prevalence $p$ is constant; each individual is infected independently with probability $p$ (the `i.i.d.\ prior'); we wish to reduce the average-case number of tests $\mathbb ET$; and we want to be \emph{certain} that each individual is correctly classified (the `zero-error' paradigm).
We emphasise the fact that $p$ is constant as $n \to \infty$ puts us in the so-called `linear regime', rather than the often-studied `sparse regime' where $p \to 0$ as $n$ gets large; the linear regime seems more relevant with applications to \covid\ and other widespread diseases.

Later, it will sometimes be convenient to instead consider the `fixed-$k$' prior, where there is a fixed number $k = pn$ of infected individuals. We discuss this mathematical convenience further in Subsection \ref{bern}.

We note also that here we are making the mathematically convenient assumption that all tests are perfectly accurate. For work that considers similar techniques to those here with imperfect tests, see, for example \cite{review, gretsky}

For more background on group testing generally, we point readers to the survey monograph \cite{survey}.

\subsection{Conservative two-stage testing}

An important distinction is between \emph{nonadaptive} testing, where all tests are designed in advance and can be carried out in parallel, and \emph{adaptive} testing, where each test result is examined before the next test pool is chosen.

Recall we are the linear regime, where $p$ is constant. For nonadaptive testing, any nonadaptive scheme using $T < n$ tests has error probability bounded away from $0$ \cite{anew} and if $T < (1 - \epsilon)n$, the error probability in fact tends to $1$ \cite{scarlett}. So simple individual testing will always be the optimal nonadaptive strategy, unless errors are tolerated -- and errors that become overwhelmingly likely as $n$ gets large. (For group testing in the linear regime with errors permitted, see, for example, \cite[Appendix B]{scarlett2}.) For adaptive testing, the best known scheme is a generalized binary splitting scheme studied by Zaman and Pippenger \cite{zaman} and Aldridge \cite{aldridge}, based on ideas of Hwang \cite{hwang}. This scheme is the optimal `nested' strategy \cite{zaman}, and is within $5\%$ of optimal for all $p \leq 1/2$ \cite{aldridge}. This algorithm (or special cases, or simplifications) was discussed in the context of \covid\ by \cite{gongalsky, golier, mentus}. However, fully adaptive schemes are unlikely to be suitable for testing for \covid\ or similar diseases, as many tests must be performed one after the other, meaning results will take a very long time to come back.

We propose, rather, using an adaptive strategy with only two stages. Within stages, tests are performed nonadaptively in parallel, so results can be returned in only the time it takes to perform two tests. This provides a good compromise between the speed but inevitable errors (or full $n$ individual tests) of nonadaptive schemes and the fewer tests but unavoidable slowness of fully adaptive schemes.
Two-stage testing goes back to the foundational work of Dorfman \cite{dorfman}, and has been discussed more recently in the context of \covid\ by \cite{aragon, benami, google, eberhardt, gongalsky, hanel, sinnott} and others.

Although our interest in two-stage testing is inspired by the usefulness in testing for \covid\ or other similar diseases, the results in this paper are theoretical, and our model is chosen for its mathematical appeal rather than direct applicability to the real world. For further discussion of models with greater direct applicability, we again direct readers to \cite{review}.

From now on, we adopt standard group testing terminology as in, for example, \cite{du-hwang, ABJ, survey, review}. In particular, individuals are `items' and infected individuals are (slightly unfortunately) `defective items'.

A two-stage algorithm that is certain to correctly classify every item works as follows:
\begin{enumerate}
    \item In the first stage, we perform some fixed number $T_1$ of nonadaptive tests. This will find some nondefective items: any item that appears in a negative tests is a \emph{definite nondefective} (DND). This will also find some defective items: any item that appears in a positive test in which every other item is DND is itself a \emph{definite defective} (DD).
    \item In the second stage, we must individually test every item whose status we so not yet know -- that is, all items except the DNDs and DDs. This requires a random number $T_2 = n - (\#\,\text{DNDs} +  \#\,\text{DDs})$ of tests.
\end{enumerate}
The total number of tests is $T = T_1 + T_2$, which has expected value $\mathbb ET = T_1 + \mathbb ET_2$

Ruling out DNDs when they appear in a negative test is a simple procedure in practice: following a negative test, a laboratory must simply report which samples were in that pool, and those items are nondefective. Further, if the test procedure can be unreliable, the procedure can easily be changed to ruling out items after they appear in some number $d > 1$ of negative tests. However, `ruling in' DDs is trickier: first, information about all the DNDs must be circulated (potentially among many different laboratories, with the privacy problems that entails), then each positive test must be carefully checked to see if all but one of the samples has been previously ruled out as a DND. Confirming that an item is defective thus involves checking a long chain of test results and pool details, which is complicated, very susceptible to occasional testing errors, and can be difficult to prove to a clinician's or patient's satisfaction.

With these problems in mind, we study the \emph{conservative two-stage group testing}\footnote{We follow \cite{survey} in using the term \emph{conservative} -- `conservative' because such an algorithm uses extra tests to ensure no false positive errors are made solely using  inference from pooled tests without an individual test. The name `trivial two-stage group testing' has also appeared more commonly in the literature -- see, for example, \cite{cheraghchi,macula,vaccaro} -- although we feel the word `trivial' might best be reserved for a variant that doesn't use the first stage at all or that performs only individual tests in the first stage.} variant. This adds the rule that every defective item must be definitively `certified' by appearing as the sole item in a (necessarily positive) test in the second stage. This gives a very simple proof that an item is defective, with a `gold standard' individual test that will not be susceptible to dilution or contamination from other samples.

So in the first stage of conservative two-stage testing, a nonadaptive scheme is used only to rule out DNDs -- that is, items that appear a negative tests and are thus definite nondefectives. However, one may not `rule in' DDs from the first stage; these must still be tested individually. Thus in the second round, each remaining item is individually tested, requiring $T_2 = n -  \#\,\text{DNDs}$ tests.

%We claim that trivial two-stage testing is likely to be useful in real-life settings. Ruling out DDs is a simple process of eliminating any items in a negative test, which is simple enough. But confirming a DD requires a positive test, then a collection of tests which prove the defectivity of each of the other items in the test. This could be very difficult if multiple laboratories work on pooled samples, and could be difficult to certify while maintaining patient confidentiality. This also relies on multiple accurate tests for a legitimate proof of defectiveness, which is more susceptible to rare incorrect test outcomes.

%It is known \cite{JAS} that when $p > n^{-1/2 + \epsilon}$, as her, then with $T_1 = (1+\delta) \frac{1}{\ln 2} pn \ln n$ we can classify \emph{every} item as a DD or DND. This usedHowever, since $T_1 > n$ for $n$ sufficiently large, this is worse than individual testing in the linear regime. Indeed, in the linear regime no adaptive test design can classify every

Note that in the first stage of non-conservative two-stage testing, we want to discover both DNDs and DDs, while in the first stage of conservative two-stage testing we can concentrate simply on discovering DNDs. Group testing experts will therefore notice that non-conservative two-stage testing has a lot in common with the `DD algorithm' of \cite{ABJ, JAS, survey}, while conservative two-stage testing is more like the `COMP algorithm' of \cite{chan, ABJ, JAS, survey}.

Two-stage testing has previously received attention in the sparse $p \to 0$ setting; we direct interested readers to \cite{mezard2011two} for more details, or \cite[Section 5.2]{survey} for a high-level overview. In the sparse regime, recovery always requires at least order $k \log n$ tests, so the difference between testing up to $k$ DDs or not -- that is, the extra number of second-stage tests required by conservative two-stage testing compared to the `non-conservative' version  -- makes up a negligible proportion of tests. Thus in the sparse regime, the distinction between conservative and non-conservative group testing is usually mathematically unimportant. It is only in the linear regime we consider here that we have to worry about the costs of definitively confirming items we think are defective and that the distinction between conservative and non-conservative testing is mathematically important. (When in this paper we give results in the limit $p \to 0$, these results apply equally well to conservative and non-conservative testing.)

\subsection{Main results}

In this paper we consider five algorithms for conservative two-stage testing. Recall that the second stage is always `test every item not ruled out as a DND', so we need only define the first stage.

\begin{description}
\item[Individual testing] tests nothing in the first stage and tests every item individually in the second stage. Although very simple, this is provably the best scheme for $p \geq (3 - \sqrt 5)/2 = 0.382$, and is the best conservative two-stage scheme of those we consider here for $p \geq 1 - 1/\sqrt[3]3 = 0.307$.
\item[Dorfman's algorithm] partitions the items into sets of size $s$ and tests each set in the first stage, then individually tests each item in the positive sets in the second stage. Dorfman's algorithm is the best scheme we consider here for $0.121 < p < 0.307$.
\item[Bernoulli first stage] where in the first stage, each item is placed in each test independently with the same probability. This scheme is suboptimal, but within $0.2$ `bits per test' of optimal for all $p$. For $p > 1/(\mathrm e + 1) = 0.269$, the optimal number of first-stage tests is $0$, and we recover individual testing.
\item[Constant tests-per-item first stage] where in the first stage, each item is placed in the same number $r$ of tests, with those tests chosen at random. This scheme is suboptimal, but very close to optimal when $p$ is small. For $p > 0.269$, the optimal number of first-stage tests is $0$, and we recover individual testing.
\item[Doubly constant first stage] where in the first stage, each item is placed in the same number $r$ of tests and each test contains the same number $s$ of items, with tests chosen at random. This is the best scheme we consider for all $p$, and is extremely close to our lower bound. For $p > 0.307$, the optimal number of first-stage tests is $0$ and we recover individual testing; while for $0.121 < p < 0.307$, the optimal number of tests per item in the first stage is $r = 1$, and we recover Dorfman's algorithm.
\item[A multi-stage scheme of Mutesa \etal] \cite{mutesa} will be used as a comparison from the literature.
\end{description}

We also give a lower bound for the number of tests required for conservative two-stage testing (Theorem \ref{lower2}). Along the way, we also find a new lower bound for usual non-conservative two-stage testing (Theorem \ref{lower1}), which may be of independent interest.

\begin{figure}[p]
    \centering
    \includegraphics[width=0.94\textwidth]{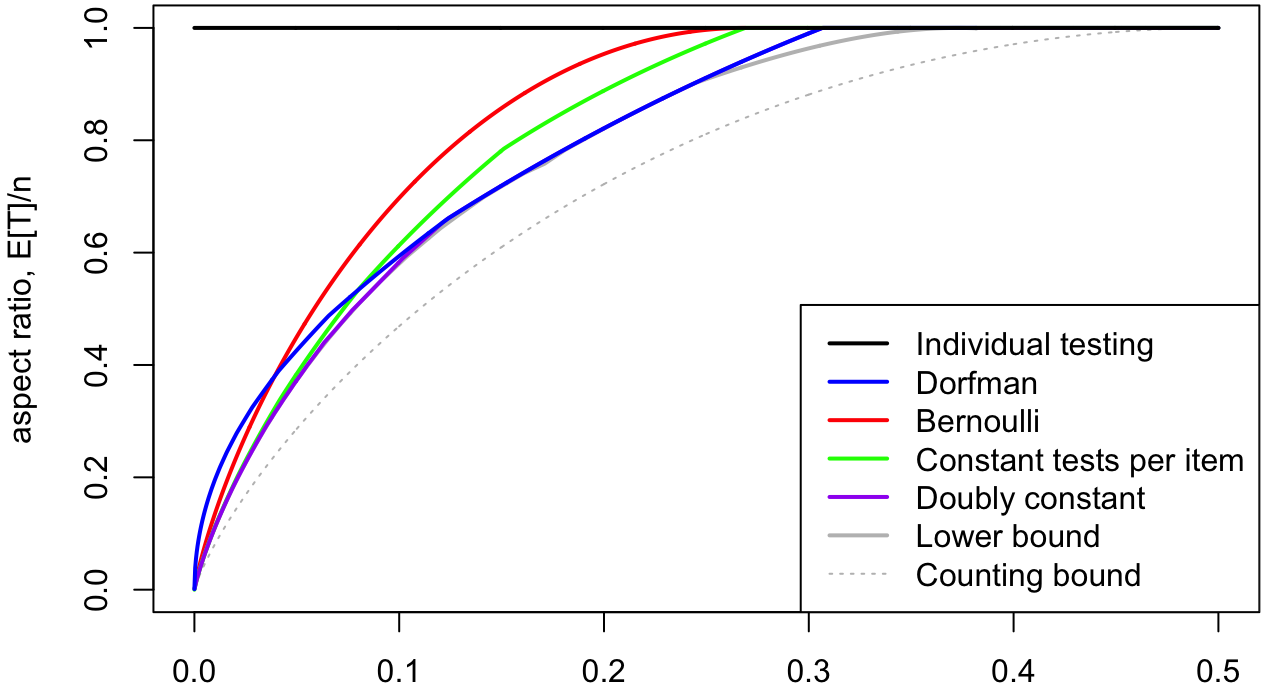} \\[8pt]
    \includegraphics[width=0.94\textwidth]{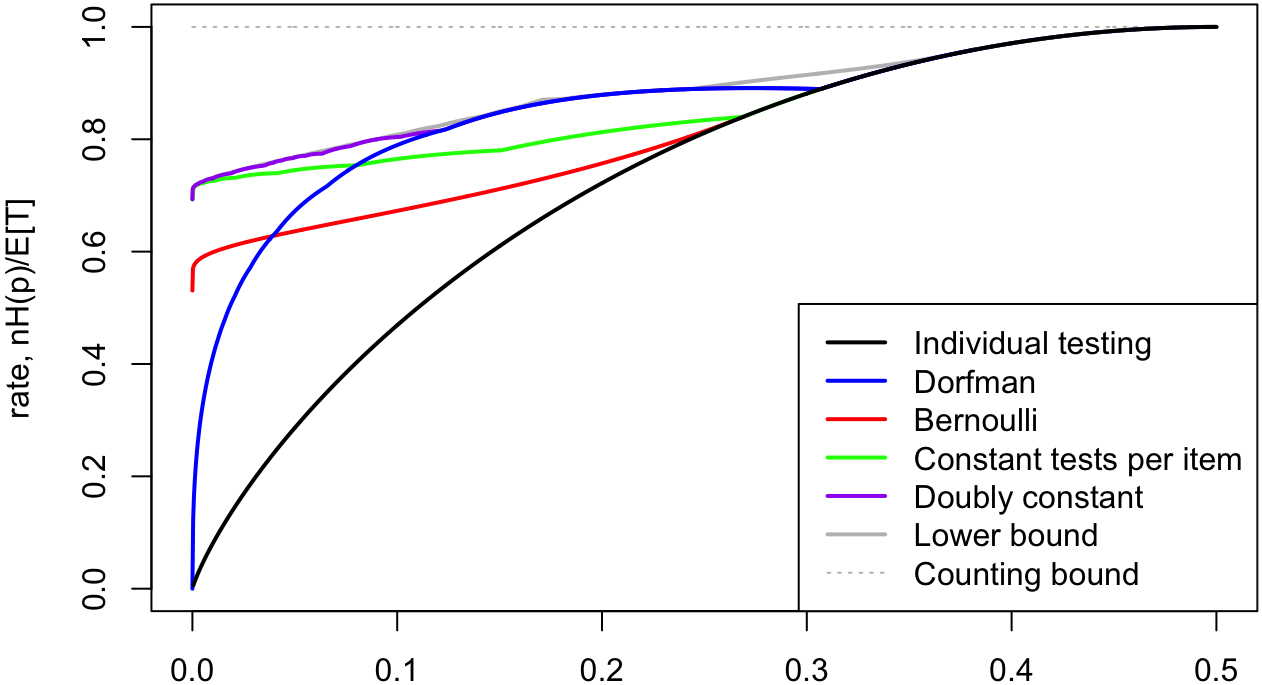} \\[8pt]
    \includegraphics[width=0.94\textwidth]{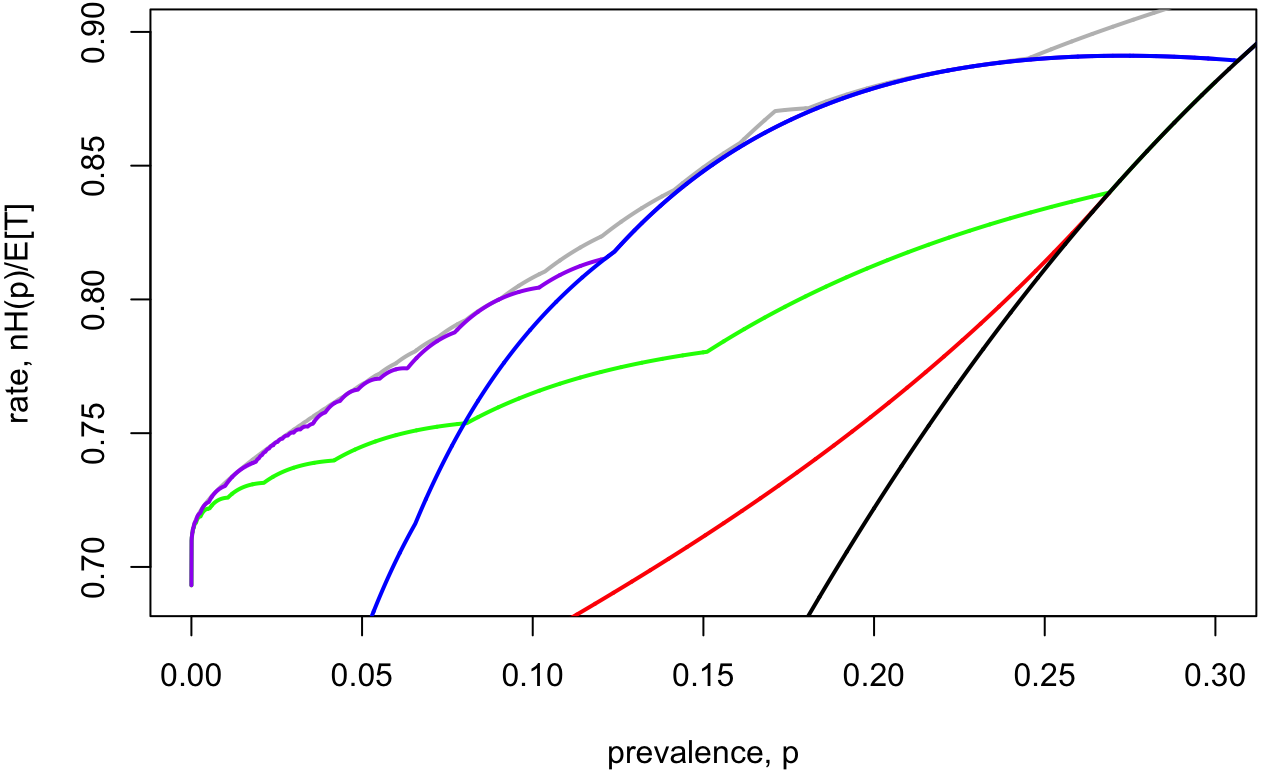}
    \caption{Theoretical performance of conservative two-stage algorithms, compared to the lower bound of Theorem \ref{lower2}.}
    \label{mainfig}
\end{figure}

Our main results on the average numbers of tests necessary are illustrated in Figure \ref{mainfig}. The top subfigure shows the \emph{aspect ratio} \cite{aldridge, review}: the expected number of tests  normalised by the number of items $\mathbb ET/n$ (lower is better) in the large $n$ limit. We compare the aspect ratio to that of individual testing with $T/n = 1$, to the \emph{counting bound} (see, for example, \cite{BJA,survey}) which says that one must have $\mathbb ET/n \geq H(p)$, where $H(p)$ is the binary entropy, and to our new lower bound of Theorem \ref{lower2}.

The middle subfigure shows the \emph{rate} $nH(p)/\mathbb ET$ (higher is better) in the large $n$ limit, which corresponds the average number of bits of information learned per test \cite{BJA, survey}. Looking at the rate allows a much clearer comparison of the different schemes when $p$ is small. The rate can be compared to individual testing, with $nH(p)/T = H(p)$ and to the counting bound that the rate is upper-bounded by $1$. Our lower bound of Theorem \ref{lower2} on the number of tests now becomes an \emph{upper} bound on the rate.

The rate of the doubly constant design is so close to the lower bound for the number of tests, that it can be difficult to see both lines on the middle subfigure. The bottom subfigure shows a zoomed in section of the rate graph, which demonstrates just how close to optimal this design is.

While the expressions in our main theorems for the expected number of tests required are smooth for fixed values of the parameters, some parameters must be integers, which means there are `jumps' in the optimal value of those parameters as they change from one integer to the next. This leads to `crooked lines' in graphs of the aspect ratio, which then corresponds to `bumpy lines' in graphs of the rate. The `kink' in the lower bound at $p = 0.171$ is where the dominant lower bound of Theorem \ref{lower1} switches from Bound 2 to Bound 3 of that theorem.

\section{Simulations}

Alongside our theoretical results for large $n$, we present evidence from simulations with $n = 1000$ items (or just above $1000$, if convenient for rounding reasons) and prevalence $p = 0.027$. We picked this value of $p$ as it corresponded to an estimate by the Imperial College \covid\ Response Team for the prevalence of \covid\ in the UK at the time these simulations were first carried out \cite{imperial}.

Specifically, we used the following algorithms, with parameters chosen according to what is optimal for $p = 0.027$ in the large-$n$ limit:
\begin{description}
    \item[Individual testing] with $n = 1000$ items, so $T = 1000$ tests.
    \item[Dorfman's algorithm] with $n = 1001$ items, and $s = 7$ items per test, so $T_1 = n/s = 143$ tests in the first stage.
    \item[Bernoulli first stage] with $n = 1000$ items, Bernoulli parameter $\pi = 1/pn = 0.037$, and $T_1 = 190$ tests in the first stage, so $\sigma = \pi n = 1/p = 37.0$ items per test on average.
    \item[Constant tests-per-item first stage] with $n = 1000$ items, $r = 4$ tests per item, and $T_1 = 160$ tests in the first stage, so $\sigma = nr/T_1 = 25$ items per test on average.
    \item[Doubly constant first stage] with $n = 1000$ items, $r = 4$ tests per item, and $s =25$ items per test, so $T_1 = nr/s = 160$ tests in the first stage.
\end{description}

We simulated each algorithm $1000$ times. 
Table \ref{simtab} shows the results of these simulations, displaying the mean number of tests used, alongside the first and ninth deciles. These simulated results are compared with a `theory' result, which takes the theoretical behaviour of $\mathbb ET$ as $n \to \infty$ (from Section \ref{secalgs}) and plugs in $n = 1000$.

\renewcommand{\arraystretch}{1.5}

\begin{table}[t]
\caption{Simulation results for conservative two-stage algorithms with $n = 1000$ and $p = 0.027$, compared to theoretical results suggested by the $n\to\infty$ limit.}
\begin{center}
\begin{tabular}{l>{\centering}p{1.6cm}>{\centering}p{1.0cm}>{\centering}p{1.0cm}>{\centering}p{1.0cm}>{\centering\arraybackslash}p{1.2cm}}
    \hline
    \multirow{2}{*}{\textbf{Algorithm}}& \multirow{2}{*}{\parbox{1.6cm}{\vspace{-1pt}\begin{center}\textbf{stage one tests}\end{center}}} & \multicolumn{4}{c}{\textbf{Total tests}} \\
     & & $10\%$ & mean & $90\%$ & Theory\\
    \hline
    Individual testing & 0 & 1000 & 1000 & 1000 & 1000\\
    Dorfman & 143 & 276 & 317.7 & 360 & 317.2 \\
    Bernoulli & 190 & 243 & 296.8 & 368 & 290.1 \\
    Constant tests-per-item & 160 & 204 & 249.7 & 302 & 243.5 \\
    Doubly constant & 160 & 205 & 245.0 & 296 & 239.3\\
    \hline
\end{tabular}
\end{center}
\label{simtab}
\end{table}

We see that, compared to individual testing, the other four algorithms each give at least a three-fold reduction in the number of tests required on average, with constant tests-per-item and doubly constant designs giving a four-fold reduction. The Bernoulli first stage was a significant improvement on Dorfman's algorithm, while the constant tests-per-item and doubly constant designs were a large improvement further. The difference between a constant tests-per-item and doubly constant first stage was small; this is not surprising, as our theoretical results show that constant tests-per-item is very close to optimal for $p$ this small (see Figure \ref{mainfig}).

We see that Dorfman's algorithm performs on average very close to theoretical predictions. The Bernoulli, constant tests-per-item and doubly constant designs require about $6$ more tests on average than the $n \to \infty$ asymptotics imply; this is presumably because the expected number of defective items $pn = 27$ is sufficiently small that rare large defective populations drive up the average number of tests in a way that becomes increasingly unlikely as $pn \to \infty$. For the numbers we used here, on $10\%$ of occasions there are at least $34$ defective items -- more than $25\%$ above the expected number of defectives -- which would require more tests than expected; but as $n \to \infty$, such a $25\%$ excesses of defective items -- even just $1\%$ excesses -- will become negligibly rare.

\section{Algorithms for conservative two-stage testing} \label{secalgs}

Throughout we write $\sim$ for asymptotic equivalence: $a(n) \sim  b(n)$ means that $a(n)/b(n) \to 1$ as $n \to \infty$, or equivalently that  $a(n) = (1 + o(1)) b(n)$; and $c(p) \sim d(p)$ means that $c(p)/d(p) \to 1$ as $p \to 0$.

\subsection{Individual testing}

Individual testing has no first round $T_1 = 0$ then tests every item in the second round $T_2 = n$. This is a conservative two-stage algorithm with $T = 0 + n = n$. This corresponds to a rate of $H(p)$, which tends to $0$ as $p \to 0$.

%It is proved in \cite{anew} that individual testing is the optimal one-stage algorithm for all $p \in (0,1)$.
It is proved in \cite{ungar} that individual testing is the optimal adaptive algorithm for all $p > (3 - \sqrt 5)/2 = 0.369$.

\subsection{Dorfman's algorithm}

Dorfman's algorithm \cite{dorfman} was the first group testing algorithm. We split the items into $n/s$ groups of size $s$. (Here $s$ has to be an integer, but since we are assuming $n$ is large we don't have to worry about $n/s$ being an integer.) If a group is positive, we test all its items individually in stage two.

Work that discusses Dorfman's algorithm in the context of testing for \covid\ includes \cite{aragon, benami, google, hanel, review}.

The following theorem is well known, but we include it here for completeness.

\begin{theorem}
Using a Dorfman first stage with with tests of size $s$, conservative two-stage testing can be completed in
\[ \mathbb E T = n \left(\frac1s + 1 - q^s\right) \]
tests on average, where $q = 1-p$.

As $p \to 0$, the optimal choice of $s$ is $s \sim 1/\sqrt{p}$, and we have $\mathbb ET \sim 2\sqrt{p}n$. The rate tends to $0$.
\end{theorem}

Dorfman's algorithm outperforms individual testing for all $p < 1 - 1/\sqrt[3]3 = 0.307$. Interestingly, Dorfman's algorithm with $s=2$ is never optimal.

\begin{proof}
Dorfman's algorithm uses $T_1 = n/s$ tests in stage 1.  Each of the $n/s$ groups is positive with probability $1 - q^s$, and if it is positive it requires $s$ more individual tests in stage 2. So the expected number of stage 2 tests is 
\[ \mathbb E T_2 = \frac{n}{s}(1 - q^s)s = (1 - q^s)n . \]
This is a total of 
\[ \mathbb ET = \frac{n}{s} + (1 - q^s)n = n\left(\frac1s + 1 - q^s\right)  \]
tests on average.

As $p \to 0$, we have
\[ \mathbb ET =n\left(\frac1s + 1 - (1-p)^s \right) \sim n \left(\frac1s + ps \right) . \]
Simple calculus shows this is  maximised at $s = 1/\sqrt{p}$, and the final part of the theorem follows.
\end{proof}

\subsection{Bernoulli first stage} \label{bern}

The Bernoulli design is the most commonly used nonadaptive design and the mathematically simplest.
In a Bernoulli design, each item is placed in each test independently with probability $\pi$. Here we suggest a Bernoulli design for the first stage of a two-stage algorithm. Bernoulli designs have been studied for nonadaptive group testing in the $p = o(1)$ regime by \cite{chan, ABJ, survey} and others. It will be convenient to write $\sigma = \pi n$ for the average number of items per test.

Although the Bernoulli first stage is not optimal (see Figure \ref{mainfig}), it is close to optimal, and the mathematical simplicity allows us to explicitly find the optimal design parameter $\pi = 1/np$ and the optimal number $T_1$ of first-stage tests. For models with slightly better performance, the design parameters can optimised be found numerically.

It will be convenient here to work here and for the following algorithms with the so-called `fixed $k$' prior, where we assume there are exactly $k = pn$ defective items, chosen uniformly at random from the $n$ items. Since we are assuming the number of items $n$ is large, standard concentration inequalities imply the true number of defectives under the i.i.d.\ prior will in fact be very close to $k = pn$. We also note that none of the algorithms we consider here will actual take advantage of exact knowledge of $k$; it is merely a mathematical convenience to make proving theorems easier. The results we prove under this `fixed $k$' prior do indeed hold for the i.i.d.\ prior also in the large $n$ limit; see \cite[Appendix to Chapter 1]{survey} for formal details of how to transfer results between the different prior models.

Throughout we write $\sim$ for asymptotic equivalence: $a(n) \sim  b(n)$ means that $a(n) = (1 + o(1)) b(n)$ as $n \to \infty$, or equivalently that $a(n)/b(n) \to 1$.

\begin{theorem}
Using a Bernoulli$(\pi)$ first stage with an average of $\sigma = \pi n$ items per test, conservative two-stage testing can be completed in
\[ \mathbb E T \sim T_1 + pn + (1-p)n \exp\left(-\sigma\mathrm{e}^{-\sigma p}\frac{T_1}{n}\right) \]
tests on average, , as $n \to infty$.

When the prevalence $p$ is known, the optimum value of $\pi$ is $1/pn$, and we can succeed with
\[ \mathbb E T \sim np\left(\mathrm e \ln\frac{1-p}{p} + 1\right)  \]
tests on average when $p \leq 1/(\mathrm e + 1) = 0.269$, or $T = n$ tests otherwise.

As $p \to 0$, we can succeed with
\[ \mathbb ET \sim \mathrm enp\ln \frac{1}{p} , \]
tests on average, and can achieve the rate $1/(\mathrm e \ln 2) = 0.531$.
\end{theorem}

\begin{proof}
We need to work out how many nondefective items are discovered by the Bernoulli design.

A given nondefective item is discovered by a test if that item is in the test but the test is negative. This happens with with probability
\[ \pi(1 - \pi)^k = \frac{\sigma}{n} \left(1 - \frac{\sigma}{n} \right)^{pn} \sim \frac{\sigma}{n} \mathrm e^{-\sigma p} . \]
When the $p$ is known, simple calculus shows that this is maximised at $\sigma = 1/p$, where it takes the value $\mathrm e^{-1}/pn$.

Thus the probability a nondefective item is not discovered is \[ \left( 1 - \frac{\sigma}{n} \mathrm e^{-\sigma p}\right)^{T_1} \sim  \exp\left(-\sigma\text{e}^{-\sigma p}\frac{T_1}{n}\right) . \]
Therefore, the total number of tests used by this algorithm on average is
\[ \mathbb E T \sim T_1 + pn + (1-p)n \exp\left(-\sigma\text{e}^{-\sigma p}\frac{T_1}{n}\right) . \]

At the optimal $\sigma = 1/p$, this is
\[ \mathbb E T \sim T_1 + pn + (1-p)n \exp\left(-\text{e}^{-1}\frac{T_1}{pn}\right) . \]
We differentiate to find the optimum value of $T_1$, giving
\[ 0 = 1 - \mathrm e^{-1} \frac{1-p}{p} \exp\left(-\text{e}^{-1}\frac{T_1}{pn}\right) , \]
from which we get the optimal value
\[ T_1 = \mathrm{e}pn \ln \left(\text{e}^{-1} \frac{1-p}{p} \right) \geq 0, \]
provided that $\text{e}^{-1} (1-p)/p \geq 1$. 
Then,
\begin{align*}
  \mathbb ET &\sim \mathrm{e}pn \ln \left(\text{e}^{-1} \frac{1-p}{p} \right) + pn + \mathrm{e}pn \\
  &= np \left( \mathrm{e} \ln \left(\text{e}^{-1} \frac{1-p}{p} \right) +1 + \mathrm e \right) \\
  &= np\left( \mathrm{e} \ln  \frac{1-p}{p}  +1  \right). 
\end{align*}
Otherwise, $T_1 = 0$ is optimal, and we have individual testing.

As $p \to 0$, we have
\[ \mathbb ET \sim np\left( \mathrm{e} \ln  \frac{1-p}{p}  +1  \right) \sim np\,\mathrm e \ln \frac{1}{p} = (\mathrm e \ln 2) n p \log_2 \frac{1}{p}. \]
Compared to the binary entropy $H(p) \sim p \log_2 1/p$, we see this gives a rate of $1/(\mathrm e \ln 2)$.
\end{proof}

\note{We can probably do similar for a non-conservative algorithm, using a similar analysis to that of the DD algorithm in \cite{ABJ}.}

\subsection{Constant tests-per-item first stage}

In a constant tests-per-item nonadaptive design, we have a constant number $r$ tests per item. For convenience, we arrange these in $r$ rounds of $T_1/r$ tests, one test per item in each round, with that tests chosen independently uniformly at random. Rounds can be conducted in parallel, so this is not adding extra stages to our two-stage algorithm. The test for an item in a given round is chosen uniformly at random from the $T_1/r$ tests, independently from other items. It will be convenient to write $\sigma = nr/T_1$ for the average number of items per test.

Constant tests-per-item designs are optimal nonadaptive designs in the sparse $p \to 0$ regime \cite{coja, JAS}, so are a good candidate for the nonadaptive stage of a two-stage scheme. It is therefore not surprising that its performance is very close to optimal when $p$ is small (see Figure \ref{mainfig}).

\begin{theorem}
Using a first stage with a constant number $r$ of tests per item and an average number of $\sigma$ items per test, conservative two-stage testing can be completed in 
\[ \mathbb ET \sim  n\left(\frac{r}{\sigma} + p + (1-p) (1 - \mathrm{e}^{-p\sigma})^r \right) \]
tests on average, as $n \to \infty$.

As $p \to 0$, we can succeed with
\[ \mathbb ET \sim \frac{1}{\ln 2}\,np\ln \frac{1}{p} , \]
tests on average, and can achieve the rate to $\ln 2 = 0.693$.
\end{theorem}

When the prevalence $p$ is known, $r$ and $\sigma$ can be numerically optimised easily.

\begin{proof}
As before, we prove our result under the fixed-$k$ prior. A nondefective item appearing in a given test sees a positive result with probability
\[ 1 - \left(1 - \frac{1}{T_1/r}\right)^k = 1 - \left(1 - \frac{\sigma}{n}\right)^{-pn} \sim 1 - \text{e}^{-p\sigma}, \]
as $n \to \infty$,
as we get a positive result unless in that round all $k$ defective items avoid the test that the given item is in.
Thus all $r$ tests are positive with probability
$(1 - \text{e}^{-p\sigma})^r$, since splitting the tests into rounds and using the fixed-$k$ prior ensures these events are independent.

Therefore the number of tests required is
\[ \mathbb ET \sim T_1 + pn + (1-p)n (1 - \text{e}^{-p\sigma})^r = n\left(\frac{r}{\sigma} + p + (1-p) (1 - \text{e}^{-p\sigma})^r \right).  \]

As $p \to 0$, we have
\[ \mathbb ET \sim n\left(\frac{r}{\sigma} + (1 - \text{e}^{-p\sigma})^r \right) . \]
Let us chose $\sigma = (\ln 2)/p$ and $r = \log_2 1/p$. (Since $r$ must be an integer, we should actually take $r$ to be the nearest integer to $\log_2 1/p$. However, since $r \to \infty$ as $p \to 0$, the rounding error is negligible in this limit.) This gives us
\begin{align*}
\mathbb ET &\sim n \left( \frac{1}{\ln 2} \,p \log_2 \frac1p + \left(\frac12\right)^{\log_2(1/p)} \right) \\
  &= n \left( \frac{1}{\ln 2} \,p \log_2 \frac1p + p \right) \\
  &\sim \frac{1}{\ln 2} np \log_2 \frac1p ,
\end{align*}
which corresponds to a rate of $\ln 2$.
\end{proof}

%The choice of $r$ and $s$ can be optimised numerically. \note{Can it be optimised explicitly by hand? Probably not?}

\begin{comment}
Number of positive tests very close to its mean
\[ M = \left(1 - \left(1 - \frac1{T_1}\right)^{rk}\right)T_1
\sim \left(1 - \exp\left(-\frac{rpn}{T_1}\right)\right)T_1 . \]

Probability a defective item is discovered
\[ 1 - \left(\frac{M}{T_1}\right)^r = 1 - \left(1 -\exp\left(-\frac{rpn}{T_1}\right)\right)^r . \]

It may be convenient to write $\rho = rpn/T_1$.
We'll probably want $\rho \approx \ln 2$, but $r$ has to be an integer.
\[ \left(1 -\exp\left(-\frac{rpn}{T_1}\right)\right)^r = (1 -\text{e}^{-\rho})^{\rho T_1/pn} .\]

In all, we get
\[ \mathbb ET = T_1 + pn + (1-p)n(1 -\text{e}^{-\rho})^{\rho T_1/pn} = T_1 + pn + (1-p)n\exp\left( -c(\rho) T_1/pn\right) , \]
where $c(\rho) = -\rho \log (1 -\text{e}^{-\rho})$.
\end{comment}

\subsection{Doubly constant first stage}

We now consider a first stage with both constant tests-per-item and constant items-per-test.
We take $r$ tests per item and $s$ items per test. Note that double-counting tells us we must have $T_1s = nr$. We again use a structure where the first stage has multiple parallel rounds. There are $r$ rounds, and each round consists of $T_1/r = n/s$ tests each containing exactly $s$ items, and each items placed in exactly $1$ test, with the design uniformly at random according to this structure. (A round can be thought of by making $n/s$ rows of $s$ boxes each, each row representing a test in that round, then placing the $n$ objects in one of those $n/s \times s = n$ boxes at random, one item per box.)

Note that $r$ and $s$ must be integers. Taking $r = 1$ and $s > 1$ recovers Dorfman's algorithm. Taking $r = 2$ gives the `double pooling' algorithm of Broder and Kumar \cite{google}. Taking $r > 2$ gives Broder and Kumar's more general `$r$-pooling' algorithm \cite{google}.

Work to discuss doubly constant designs in the context of testing for \covid\ includes \cite{benami, google, sinnott}.

\note{Q: Is it interesting to allow `$s$ or $s+1$' tests per item, if the `optimal' choice seems to be $s$-and-a-half, say?}

\begin{theorem}
Using a first stage with a constant number $r$ of tests per item and a constant number $s$ of items per test, conservative two-stage testing can be completed in 
\[ \mathbb ET \sim  n\left(\frac{r}{s} + p + q (1 - q^{s-1})^r \right) \]
tests on average, where $q = 1 - p$, as $n \to \infty$.

As $p \to 0$, we can succeed with
\[ \mathbb ET \sim \frac{1}{\ln 2}\,np\ln \frac{1}{p} , \]
tests on average, and can achieve the rate to $\ln 2 = 0.693$.
\end{theorem}

When the prevalence $p$ is known, $r$ and $s$ can be numerically optimised.

Note that putting $r = 1$ does indeed give
\[ \mathbb ET \sim  n \left( \frac 1s + p + q(1 - q^{s-1})\right) = n \left( \frac 1s + 1 - q^s \right)  , \]
as for Dorfman's algorithm.

We note that the expression here is the same as that heuristically demonstrated by Broder and Kumar \cite{google}, who use the i.i.d\ prior as if there were independence within rounds. (They say that they will discuss the accuracy of this approximation in `the final version' of \cite{google}). By using the fixed-$k$ prior here, we actually do have independence within rounds, so can formally prove the result. In the large $n$ limit, this then transfers to the i.i.d.\ prior, as discussed earlier and in \cite[Appendix to Chapter 1]{survey}.

\begin{proof}
Fix $r$ and $s$. Consider a test containing a given nondefective item. There are $\binom{n-1}{s-1}$ ways for the remaining $s-1$ items in the test to be filled, and $\binom{n-1-k}{s-1}$ ways for it to be filled up with other nondefective items. Therefore, 
the probability the probability that the test is negative is
\[ \frac{\binom{n-k-1}{s-1}}{\binom{n-1}{s-1}} \sim \left( \frac{n-k-1}{n-1} \right)^{s-1} =  \left( 1-\frac{k}{n-1} \right)^{s-1} \sim (1-p)^{s-1} = q^{s-1} , \]
Since with the fixed-$k$ prior we have independence between rounds, we have that the probability all the tests containing the nondefective item are positive -- and this that the nondefective item requires retesting in the second stage -- is
$(1 - q^{s-1})^r$.  All $k = pn$ defective items must be retested in the second stage, of course.

Over all, the expected number of tests required is
\[ \mathbb ET \sim T_1 + pn + qn(1 - q^{s-1})^r = n \left( \frac rs + p + q(1 - q^{s-1})^r \right) . \]

The analysis as $p \to 0$ is the same as for the constant tests-per-item design: we again take $s = (\ln 2)/p$ and $r = \log_2 1/p$ (or the nearest integers thereto), and the asymptotic behaviour is identical in the $p \to 0$ limit.
\end{proof}

\note{One would think picking $s$ so tests are close to $50:50$ positive and negative would be good.}

%This is the same expression found by Broder and Kumar \cite{google}. In their algorithm, they further structure the tests by having (in our notation) $s$ rounds of test; in each round, each item takes part in exactly one of the $n/m = T/s$ tests. These rounds can be conducted in parallel, so count as one `stage' in our terminology. The above argument shows that this extra structure does not improve the average number of tests required for large $n$.

\subsection{Comparison to a multi-stage scheme of Mutesa \etal}

We will briefly compare the results of our schemes with a scheme from a \emph{Nature} paper by Mutesa \etal\ \cite{mutesa}. This scheme is multiple-stage scheme that `usually` requires two stages, and uses similar ideas to ours in this paper. However, Mutesa \etal\ report that their scheme requires $\mathrm ep\ln(1/p)$ tests as $p \to 0$, which corresponds to a rate of $1/(\mathrm e\ln2) = 0.531$, which is inferior to the rate of the constant tests-per-item and doubly-constant designs we looked at here.

The scheme of \cite{mutesa} works as follows:
\begin{enumerate}
\item The first stage uses a Dorfman-like design, where each item is tested once in a test of $s_1$ items. This parameter $s_1$ is chosen to be of the form $s_1 = a^{r_2}$ for positive integers $a$ and $r_2$. If a test is negative, all $s_1$ items can be definitively ruled as nondefective; if a test is positive, the $s_1$ items go through to stage 2. 
\item The second stage deals with $s_1$ items from stage 2 using a doubly constant design, with $r_2$ tests per item and $s_2 = a^{r_2 - 1}$ items per test.
\end{enumerate}

If a set of $s_1$ items contains no defective items, that is discovered in the first stage. If the set contains exactly one defective item, that is discovered in the second stage, albeit non-conservatively, without the benefit of a `gold standard' individual test, as our schemes in this paper always have. If the set contains two or more defective items, it will typically not be possible to identify them (although one can get lucky), and further stages of testing will be needed -- whereas the schemes in this paper are always guaranteed to complete in two stages.

The second stage of the scheme of Mutesa \etal\ does not, as we do here, use a random design subject to the tests-per-item and items-per-test constraints. Rather the parameters are chosen to allow a `hypercube' design. The $s_1 = a^{r_2}$ items can be pictures as making up an $r_2$-dimensional $a \times a \times \cdots \times a$ hypercube. Then each of the $T_2 = ar_2$ tests represents an $(r-1)$-dimensional slice of the hypercube. We don't believe this changes the mathematics in the $n \to \infty$ limit compared to a fully randomised design like those we consider here, but it may have practical benefits in terms of the simplicity of running the tests in the laboratory and may have better performance at finite $n$, although this restricts the choices of parameters $s_1, r_2, s_2$, which can no longer be arbitrary integers.

\section{Lower bounds}

In order to see how good our conservative two-stage algorithms are, we will compare the number of tests they require to a theoretical lower bound (Theorem~\ref{lower2}).

It will be convenient to start with a lower bound for usual non-conservative two-stage testing (Theorem \ref{lower1}), which may be of independent interest. We will then show how to adapt the argument to conservative two-stage testing.

\note{could be more precise here about where these bounds are dominated by counting bound and/or Riccio--Colbourn}

\subsection{Lower bound for two-stage testing}

Let us start by thinking about a lower bound on the number of tests necessary for usual two-stage testing. 

\begin{theorem} \label{lower1}
The expected number of tests required for two-stage testing is at least
  \[ \mathbb E T \geq n \frac{1}{f(p)} \ln \frac{1}{f(p)} + n \exp \left(\ln \frac{1}{f(p)}\right) = n \frac{1}{f(p)}\left(\ln \frac{1}{f(p)} + 1\right) ,\]
where
\[ f(p) = \max_{w=2,3,\dots} \left\{-w \ln\big(1 - (1-p)^{w-1}\big) \right\} . \]
\end{theorem}

\begin{proof}
Our goal is to bound the expected number $T_2$ of items that are not classified as DND or DD. 

A nondefective item fails to be classified DND if and only if it only appears in positive tests -- that is, if for each test it is in, one of the other items is defective. A defective item fails to be classified DD if -- but not only if -- for each test it is in, one of the other items is defective. (It's not `only if' because finding a DD requires one of its tests to contain solely \emph{definite} nondefectives, but we are only seeking a bound.) \note{This seems very likely to be the loosest part of the bound. How can we improve it?} Let us call an item \emph{hidden} if every test it is in contains at least one other defective item. Then
\[ \mathbb ET_2 \geq \mathbb E (\#\text{ hidden items}) = \sum_{i=1}^n \mathbb P(H_i) , \]
where $H_i$ is the event that item $i$ is hidden.

We seek a bound at least as good as individual testing $T = n$. Then without loss of generality we may assume there are no tests of weight $w_t=1$ in the first stage. If there is one, remove it and the item it tests; this leaves $p$ the same, does not increase the error probability, and reduces the number of available tests per item.

It will be convenient to write $x_{ti} = 1$ if item $i$ is in test $t$, and $x_{ti} = 0$ if it is not. With this notation, the probability that item $i$ is hidden is bounded by 
  \begin{equation} \label{lemeq}
    \mathbb P(H_i) \geq \prod_{t : x_{ti} = 1} (1 - q^{w_t-1}) \, ,
  \end{equation}
where $q = 1 - p$, due to a result of \cite{anew}. Note that $1 - q^{w_t - 1}$ is the probability of the event that $i$ gets hidden in test $t$, and the bound \eqref{lemeq} follows by applying the FKG inequality to these increasing events; see \cite{anew} for details.
  
It will be useful later to write $L(i)$ for the logarithm of the bound \eqref{lemeq}, so $\mathbb P(H_i) \geq \mathrm e^{L(i)}$, where
  \begin{align*}
    L(i) &= \ln \prod_{t : x_{ti} = 1} (1 - q^{w_t-1}) \\
      &= \sum_{t : x_{ti} = 1} \ln(1 - q^{w_t-1}) \\
      & = \sum_{t=1}^T x_{ti} \ln(1 - q^{w_t-1}) .
  \end{align*}
The expected number of hidden items is
  \[ \mathbb E T_2 = \sum_{i=1}^n \mathbb P(H_i) = \sum_{i=1}^n \mathrm{e}^{L(i)} . \]
  
We now use the arithmetic mean--geometric mean inequality in the form
  \[ \sum_{i=1}^n \mathrm{e}^{a_i} \geq n \exp \left( \frac1n \sum_{i=1}^n a_i \right) , \]
to get the bound
  \begin{equation} \label{expeq}
    \mathbb E T_2 \geq n \exp \left( \frac 1n \sum_{i=1}^n L(i) \right) .
  \end{equation}
We now need to bound term inside the exponential.

By manipulations similar to those in \cite{anew} we have
\begin{align*}
  \frac 1n \sum_{i=1}^n L(i) &= 
  \frac1n \sum_{i=1}^n L(i) \\
    &= \frac1n \sum_{i=1}^n \sum_{t=1}^{T_1} x_{ti} \ln(1 - q^{w_t-1}) \\
    &= \frac1n \sum_{t=1}^{T_1} \left( \sum_{i=1}^n x_{ti}\right) \ln(1 - q^{w_t-1}) \\
    &= \frac1n \sum_{t=1}^{T_1} w_t \ln(1 - q^{w_t-1}) \\
    &\geq \frac{1}{n} \,T_1 \,\min_{t=1,2,\dots,T_1} \left\{w_t \ln(1 - q^{w_t-1}) \right\} \\
    &\geq \frac{T_1}{n}\,\min_{w=2,3,\dots,n} \left\{w \ln(1 - q^{w-1}) \right\}\\
    &\geq  -f(p) \,\frac{T_1}{n},
\end{align*}
where
\[ f(p) = -\min_{w=2,3,\dots} \left\{w \ln\big(1 - (1-p)^{w-1}\big) \right\} = \max_{w=2,3,\dots} \left\{-w \ln\big(1 - (1-p)^{w-1}\big) \right\} , \]
as in the statement of the theorem.
(We introduce the minus sign so that $f(p)$ is positive.)

Putting this back into \eqref{expeq}, we get
  \[ \mathbb E T_2 \geq n \exp \left(-f(p) \frac{T_1}{n}\right) .\]
Thus the total expected number of tests required is at least
  \[ \mathbb E T = T_1 + \mathbb E T_2 \geq T_1 + n \exp \left(-f(p) \frac{T_1}{n}\right) .\]
  
To find the optimal value of $T_1$, we differentiate this, to get
\[ 0 = 1 - f(p) \exp \left(-f(p) \frac{T_1}{n}\right) , \]
with optimum
\[ T_1 = n \frac{1}{f(p)} \ln f(p) . \]
Thus
  \[ \mathbb E T \geq n \frac{1}{f(p)} \ln f(p) + n \exp \left(-\ln f(p)\right) = n \frac{1}{f(p)}\left(\ln f(p) + 1\right) ,\]
and we are done.
\end{proof}

\subsection{Lower bound for conservative two-stage testing}

We can now use the machinery of the previous result to prove a lower bound for conservative two-stage testing.

\begin{theorem} \label{lower2}
For conservative two-stage group testing we have the following bounds:
\begin{enumerate}
    \item $\mathbb ET \geq n$ for $p \geq (3 - \sqrt{5})/2 = 0.382$;
    \item $\mathbb E T \geq n \displaystyle\frac{1}{g(p)}\left(\ln g(p) + 1\right)$;
    \item ${\displaystyle \mathbb E T \geq n \left(p + \frac{1}{f(p)} \Big(\ln \big((1-p)f(p)\big) + 1 \Big)\right)}$.
\end{enumerate}
where
\begin{align*}
    f(p) &= \max_{w=2,3,\dots} \left\{-w \ln\big(1 - (1-p)^{w-1}\big) \right\} \\
    g(p) &= \max_{w=2,3,\dots} \left\{-w \ln\big(1 - (1-p)^{w\phantom{{}-1}}\big) \right\} .
\end{align*}

In the limit as $p \to 0$, this gives an upper bound on the rate of $\ln 2 = 0.693$.
\end{theorem}

Here, $f$ is the same as in Theorem \ref{lower1}. It may be useful to note that Bound 2 dominates for $p < 0.171$, and Bound 3 dominates for $0.171 < p < 0.382$.

\note{That we get out both Bounds 2 and 3 from essentially the same method, yet they're both tight for different ranges of $p$, has my spidey-sense tingling, but the argument seems good to me.}

 %A useful approximation to $g(p)$ is given by
%\[ \bar g(p) \max_{w \in \mathbb R} \left\{-w \ln\big(1 - (1-p)^{w}\big) \right\} , \]
%where $w$ is relaxed to be any real number. The optimal $w$ is $-(\ln 2)/\ln (1-p)$, giving
%\[ \bar g(p) = \frac{(\ln 2)^2}{\ln (1-p)} = \frac{ \ln 2}{\log_2 (1-p)}  . \]

Comparing these bounds with the results of our algorithms (see Figure \ref{mainfig}), we see that testing with a doubly constant first stage is very close to optimal for all $p$, and \emph{extremely} close to optimal for $p \leq 1/4$. Comparing this result with the theorems from the previous section, we see that the constant tests-per-item and doubly constant designs achieve the optimal rate in the limit $p \to 0$.

\begin{proof}
Bound 1 is a universal bound of Ungar \cite{ungar} that applies to any group testing algorithm. It's left to prove Bounds 2 and 3.

The proof of the bound for conservative two-stage testing proceeds in a similar way to that of Theorem \ref{lower1}. There are two different ways we can count the number of items that require testing in the second stage. For Bound 2, we count every item that appears solely in positive tests -- such an item is either defective or a hidden nondefective. For Bound 3, we count all the defective items, of which there are $pn$ on average, plus all \emph{nondefective} items that appear solely in positive tests. 

For Bound 2, the probability a test of weight $w$ is positive is $1 - q^w$, where $q = 1- p$. We use the same argument as before -- this time in less detail. (For conservative two-stage testing we don't have to be so careful about ruling out individual tests in the first round: they can simply be moved into the second round.) The probability an item $i$ appears in only positive tests is
$\mathbb P(E_i) \geq \prod_{t : x_{ti} = 1} (1 - q^{w_t})$,
by the FKG inequality. Going through exactly the same argument, the expected number of items in only positive tests is
\[ \mathbb E T_2 \geq n \exp \left(-g(p) \frac{T_1}{n}\right) , \]
where
\[ g(p) = \max_{w=2,3,\dots} \left\{-w \ln\big(1 - (1-p)^{w}\big) \right\} , \]
giving an average number of tests
\[ \mathbb ET \geq T_1 + n \exp \left(-g(p) \frac{T_1}{n}\right) . \]
Optimising $T_1$ the same way gives the final bound
  \[ \mathbb E T \geq = n \frac{1}{g(p)}\left(\ln g(p) + 1\right) .\]
  
For Bound 3, we must test the average of $pn$ defective items, plus the average of $(1-p)n \mathbb P(H_i)$ hidden nondefectives; here $(1 - p)n$ is the average number of nondefectives, and $\mathbb P(H_i)$ is the probability a given nondefective is hidden. We can use from before the bound
\[ \mathbb P(H_i) \geq \exp \left(-f(p) \frac{T_1}{n}\right) . \]
This gives
\[ \mathbb ET \geq T_1 + pn + (1-p)n \exp \left(-f(p) \frac{T_1}{n}\right) . \]
Optimising in the same way gives
\[ T_1 = n \frac{1}{f(p)} \ln \big((1-p)f(p) \big) , \]
and hence
\begin{align*} \mathbb ET & \geq n \frac{1}{f(p)} \ln \big((1-p)f(p) \big) + pn + (1-p)n \frac{1}{(1-p)f(p)} \\
&= n \left(p + \frac{1}{f(p)} \Big(\ln \big((1-p)f(p)\big) + 1 \Big)\right) . \end{align*}

For the result as $p \to 0$, we concentrate on Bound 2. In the maximum expression for $g(p)$ we take $w = (\ln 2)/p$, to get
\begin{align*}
g(p) &\leq - \frac{\ln 2}{p} \ln \big(1 - (1-p)^{(\ln2)/p} \big) 
     \sim - \frac{\ln 2}{p} \ln \big(1 - \mathrm e^{-\ln 2} \big) 
     = \frac{(\ln 2)^2}{p} .
\end{align*}
Then Bound 2 becomes
\begin{align*}
\mathbb ET &\geq n \,\frac{p}{(\ln 2)^2} \left(\ln \left(\frac{(\ln 2)^2}{p}\right) + 1\right) \\
&\sim n \,\frac{p}{(\ln 2)^2} \ln \frac{1}{p} \\
&= \frac{1}{\ln 2} \, np \log_2 \frac 1p .
\end{align*}
Comparing this with $H(p) \sim p \log_2 (1/p)$, we see we have an upper bound of $\ln 2$ on the rate.
\end{proof}

\section*{Acknowledgements}

The author was supported in part by UKRI Research Grant EP/W000032/1.

The author thanks David Ellis, Oliver Johnson, and Jonathan Scarlett for helpful discussions, and  Mahdi Cheraghchi, Anthony Macula, and Ugo Vaccaro for useful pointers to relevant literature.

\bibliographystyle{abbrv}
\bibliography{bibliography}

\end{document}